\documentclass[a4paper,11pt]{article}
\usepackage[left=3.17cm,right=3.17cm,top=2.54cm,
headheight=0.5cm,headsep=0.54cm,bottom=2.54cm,footskip=0.79cm
]{geometry}

\usepackage[all]{xy}
\usepackage{amsmath,amsfonts,amssymb,amsthm,epsfig,amscd,comment,latexsym,psfrag}
\usepackage{nicefrac,xspace,tikz}

\usetikzlibrary{arrows}
\usetikzlibrary{decorations.markings}
\def\Z{\mathbb Z}

\def\Q{\mathbb Q}

\def\C{\mathbb C}

\def\1{{\bf{1}}}

\usepackage[pdfstartview=FitH]{hyperref}

\def\footnoterule{\kern 1mm \hrule width 7cm \kern 2.2mm}%

 \newtheorem{thm}{Theorem}[section]
 \newtheorem{prp}[thm]{Proposition}
 \newtheorem{lem}[thm]{Lemma}
 \newtheorem{cor}[thm]{Corollary}

\newcommand{\bea}{\begin{eqnarray}}
\newcommand{\eea}{\end{eqnarray}}
\newcommand{\be}{\begin{equation}}
\newcommand{\ee}{\end{equation}}

\begin{document}

\title{Coupled Hall-Littlewood functions, vertex operators and the $q$-boson model}
\author{
\  \ Na Wang\dag, Chuanzhong Li\ddag\footnote{Corresponding author's email:lichuanzhong@nbu.edu.cn}\\
\dag\small Department of mathematics and statistics, Henan University, Kaifeng, 475001, China.\\
\ddag\small Department of Mathematics,  Ningbo University, Ningbo, 315211, China}

\maketitle
\thanks{}
\date{}
\begin{abstract}
In this paper, we firstly give the definition of the coupled Hall-Littlewood function and its realization in terms of vertex operators. Then we construct the representation of two-site generalized $q$-boson model in the algebra of coupled Hall-Littlewood functions. Finally, we find that the vertex operators which generate coupled Hall-Littlewood functions can also be used to obtain the  partition function of the A-model topological string on the conifold.

\end{abstract}
\noindent

\noindent Mathematics Subject Classifications (2010).  37K05, 37K10, 37K20, 17B65, 17B67.\\
{\bf Keywords: }{coupled Hall-Littlewood functions, vertex operators, $q$-boson model, topological strings.}

\section{Introduction}\label{sec1}
Symmetric functions have played an important role in mathematics (representation theory, combinatorics) for a long time\cite{Mac,stan}.
Symmetric functions also appear in mathematical physics, especially in integrable models. Kyoto school use Schur functions in a remarkable way to understand the KP and KdV hierarchies\cite{MJD}. T. Tsuda defined the Universal Character (UC) hierarchy (a generalization of KP hierarchy) and obtained that the tau functions of UC hierarchy can be realized in terms of universal characters. He also proved that the UC hierarchy has relations with Painlev\'e equations by the similar reductions\cite{Tsuda}. In this paper, we consider two different subjects: the algebra of coupled Hall-Littlewood functions and the $q$-boson model.

The first purpose of this paper is to give the definition of coupled Hall-Littlewood functions. Schur functions and Hall-Littlewood functions are well known famous symmetric functions. The Hall-Littlewood function $Q_\lambda({\bf x})$ is a generalization of the Schur function $S_\lambda({\bf x})$, while the universal character $S_{[\lambda,\mu]}({\bf x},{\bf y})$ is a generalization of the Schur function $S_\lambda({\bf x})$ in another way.
\[\begin{tikzpicture}[scale=.75][>=stealth]

\draw (0,2) node {$S_\lambda({\bf x})$};
\draw (0,0) node {$S_{[\lambda,\mu]}({\bf x},{\bf y})$};
\draw (6,2) node {$Q_\lambda({\bf x})$};
\draw (6,0) node {$Q_{[\lambda,\mu]}({\bf x},{\bf y})$};
\draw (1.5,0) -- (4.5,0) [<-];
\draw (1.5,2) -- (4.5,2) [<-];
\draw (0,0.5) -- (0,1.5) [->];
\draw (6,0.5) -- (6,1.5) [->];
\draw (3,0.3) node {$t=0$};
\draw (3,2.3) node {$t=0$};
\draw (0.8,1) node {$\mu=\emptyset$};
\draw (6.8,1) node {$\mu=\emptyset$};
\end{tikzpicture}.\]
In this paper, we define the coupled Hall-Littlewood function $Q_{[\lambda,\mu]}({\bf x},{\bf y})$ and construct its vertex operator realization.

The second purpose of this paper is to give the representation of the two-site generalized $q$-boson model and the Quantum Inverse Scattering Method in the algebra of coupled Hall-Littlewood functions $Q_{[\lambda,\mu]}({\bf x},{\bf y})$. The $q$-boson model is a strongly correlated one-dimensional boson system on a finite chain, which is an integrable model and can be solved in the formulism of the Quantum Inverse Scattering Method \cite{KBI}. The $q$-boson model is important in several branches of modern physics, such as solid state physics and quantum nonlinear optics. In this paper, we will give that the two-site generalized $q$-boson model and the Quantum Inverse Scattering Method have an interpretation in terms of the algebra of coupled Hall-Littlewood functions.

The relation between the algebra of Hall-Littlewood functions and the $q$-boson model is known from \cite{PS,NVT}. The map from the states in the $q$-boson model to the Hall-Littlewood functions is given by
\[
\bigotimes_{i=0}^{M} |n_i\rangle_i\mapsto Q_{\lambda}({\bf x}),\quad \lambda=1^{n_1}2^{n_2}\ldots.
\]
Under this realization, the creation operator $B(u)$ coincides with
\be
H_M({\bf x},u^2)=1+q_1u^2+q_2u^4+\cdots q_Mu^{2M}
\ee
where ${\bf x}=(x_1,x_2,\cdots)$ and $q_k$ is the Hall-Littlewood function $Q_{(k)}({\bf x})$. The annihilation operator $C(u)$ is the adjoint operator
\be
H_M^\bot({\bf x},u^{-2})=1+q_1^\bot u^{-2}+q_2^\bot u^{-4}+\cdots q_M^\bot u^{-2M}
\ee
with respect to the standard scalar product in the algebra of Hall-Littlewood functions.

In this paper, we generalize the results in \cite{PS,NVT} to give the realization of two-site generalized $q$-boson model in the algebra of coupled Hall-Littlewood functions.
We define the map from the states in the two-site generalized $q$-boson model to the coupled Hall-Littlewood functions by
\[
\bigotimes_{i=0}^{M_1} |n_i\rangle_i^{(1)}\bigotimes\bigotimes_{i=0}^{M_2} |m_i\rangle_i^{(2)}\mapsto Q_{[\lambda,\mu]}({\bf x},{\bf y})
\]
with
\[
\lambda=1^{n_1}2^{n_2}\ldots,\quad \mu=1^{m_1}2^{m_2}\ldots.
\]
We find that in monodromy matrix
\bea
&&T(u)=\left( \begin{array}{cc}
A_2(u) & B_2(u) \\
C_2(u) & D_2(u)
\end{array} \right)\left( \begin{array}{cc}
A_1(u) & B_1(u) \\
C_1(u) & D_1(u)
\end{array} \right)\nonumber
\eea
the entries $B_1(u)$ and $B_2(u)$ coincide with
\[
 H_{M_1}({\bf x}-\tilde{\partial}_{\bf y},u^2)\ \text{and}\  H_{M_2}({\bf y}-\tilde{\partial}_{\bf x},u^2)
\]
respectively,
where $$\tilde{\partial}_{\bf x}=(\partial_{x_1}, \frac{1}{2}\partial_{x_2},\frac{1}{3}\partial_{x_3},\cdots), \quad \partial_{x_i}=\frac{\partial}{\partial{x_i}}.$$
The entries $C_1(u)$ and $C_2(u)$ coincide with
\[
 H_{M_1}^\bot({\bf x}-\tilde{\partial}_{\bf y},u^{-2})\  \text{and}\  H_{M_2}^\bot({\bf y}-\tilde{\partial}_{\bf x},u^{-2})
\]
respectively. Therefore, in the $M_1,M_2\rightarrow \infty$ limit, the entries  in monodromy matrix $T(u)$ have relations with
the vertex operators
\bea
&&\Gamma_1^-(z)=e^{\xi_t({\bf x}-\tilde{\partial}_{\bf y}, z)},\quad \Gamma_1^+(z)=e^{\xi(\tilde{\partial}_{\bf x}, z^{-1})},\\
&&\Gamma_2^-(z)=e^{\xi_t({\bf y}-\tilde{\partial}_{\bf x}, z)},\quad \Gamma_2^+(z)=e^{\xi(\tilde{\partial}_{\bf y}, z^{-1})},
\eea
where $\xi_t({\bf x},z)=\sum_{n=1}^\infty (1-t^n)x_n z^n,\ \xi({\bf x},z)=\xi_0({\bf x},z)$.
These vertex operators are the ones who generate the coupled Hall-Littlewood functions and the deformed fermions. From these results, we can get the expression of $|\Psi_N(u_1,\cdots,u_N)\rangle$ in the algebra of coupled Hall-Littlewood functions.

The third purpose of this paper is to discuss the relations between the vertex operators $\Gamma_i^-(z)$, $\Gamma_i^+(z)$ $(i=1,2)$ and the A-model topological string partition function $Z_{conifold}^{top}(z,t)$ on the conifold. It is known that the A-model topological string partition function $Z_{conifold}^{top}(z,t)$ on the conifold can be written as a fermionic correlator involving the vertex operators $\tilde{\Gamma}^-(z)=e^{\xi_t({\bf x},z)},\quad \tilde{\Gamma}^+(z)=e^{\xi(\tilde{\partial}_{\bf x},z^{-1})}$ with a particular specialization of the values of $z=q^{\pm 1/2},q^{\pm 3/2},q^{\pm 5/2}\cdots$. In this paper, we will give that $Z_{conifold}^{top}(z,t)$ can also be obtained from the vertex operators $\Gamma_i^-(z)$ and $\Gamma_i^+(z)$ $(i=1,2)$  with the same specialization of the values of $z$.

The paper is organized as follows. In section \ref{sect2}, we give the definition of coupled Hall-Littlewood functions and construct its vertex operator realization. In section \ref{sect3}, we recall the $q$-boson model. In section \ref{sect4}, we give the representation of the two-site generalized $q$-boson model in the algebra of coupled Hall-Littlewood functions, and we find that the actions of the entries in monodromy matrix on coupled Hall-Littlewood functions are obtained from the truncated expansions of the vertex operators discussed in section \ref{sect2}. In section \ref{sect5}, we give that the A-model topological string partition function $Z_{conifold}^{top}(z,t)$ on the conifold can be obtained from these vertex operators.
\section{Coupled Hall-Littlewood functions and vertex operators}\label{sect2}
In this section, we will give the definition of coupled Hall-Littlewood function, and construct its vertex operator realization.
Let ${\bf x}=(x_1,x_2,\cdots)$ and ${\bf y}=(y_1,y_2,\cdots)$.
The operators $q_n({\bf x})$ are determined by the generated function\cite{Jing,FW}:
\be\label{qr}
\sum_{n=0}^\infty q_n({\bf x})z^n=e^{\xi_t({\bf x},z)},\quad \xi_t({\bf x},z)=\sum_{n=1}^\infty(1-t^n) x_n z^n
\ee
and set $q_n({\bf x})=0$ for $n<0$. When $t=0$, the operators $q_n({\bf x})$ turn into the complete homogeneous symmetric function $h_n({\bf x})$ if we replace $ix_i$ with the power sum $p_i$.

For a pair of Young diagrams $\lambda=(\lambda_1,\lambda_2,\cdots,\lambda_l)$ and $\mu=(\mu_1,\mu_2,\cdots,\mu_{l'})$, let
\be
q_{[\lambda,\mu]}({\bf x},{\bf y})=q_{\lambda_1}({\bf x})\cdots q_{\lambda_l}({\bf x})q_{\mu_1}({\bf y})\cdots q_{\mu_{l'}}({\bf y}).
\ee
From the property of $q_\lambda$ discussed in \cite{Mac}, we have that $q_{[\lambda,\mu]}({\bf x},{\bf y})$ form a basis of $\Q(t)[{\bf x},{\bf y}]=\Q(t)\otimes\Q[{\bf x},{\bf y}]$.
The raising operator $R_{ij}$ is defined by
\[
R_{ij}\cdot\lambda=(\lambda_1,\cdots,\lambda_i+1,\cdots,\lambda_j-1,\cdots, \lambda_l).
\]
We define the operators $R_{\lambda_i \lambda_j}$ and $R_{\mu_i \mu_j}$ by
\be
R_{\lambda_i\lambda_j}\cdot q_{[\lambda,\mu]}({\bf x},{\bf y})=q_{[R_{ij}\cdot\lambda,\mu]}({\bf x},{\bf y}),\quad R_{\mu_i \mu_j}\cdot q_{[\lambda,\mu]}({\bf x},{\bf y})=q_{[\lambda,R_{ij}\cdot\mu]}({\bf x},{\bf y})
\ee
and define the operator $D_{\lambda_i \mu_j}$ by
\be
D_{\lambda_i \mu_j}\cdot q_{[\lambda,\mu]}({\bf x},{\bf y})=q_{[\tilde{\lambda},\tilde{\mu}]}({\bf x},{\bf y})
\ee
where $\tilde{\lambda}=(\lambda_1,\cdots,\lambda_i-1,\cdots,\lambda_l)$ and $\tilde{\mu}=(\mu_1,\cdots,\mu_j-1,\cdots,\mu_l)$.

For $\lambda=(\lambda_1,\lambda_2,\cdots,\lambda_l)$ and $\mu=(\mu_1,\mu_2,\cdots,\mu_{l'})$, we define the coupled Hall-Littlewood function $Q_{[\lambda,\mu]}=Q_{[\lambda,\mu]}({\bf x},{\bf y})$ as a polynomial of variables ${\bf x}$ and ${\bf y}$ in $\Q(t)[{\bf x},{\bf y}]$:
\be\label{Qdef}
Q_{[\lambda,\mu]}({\bf x},{\bf y})=\prod_{ i<j}\frac{1-R_{\lambda_i\lambda_j}}{1-tR_{\lambda_i\lambda_j}}\prod_{a<b}\frac{1-R_{\mu_a\mu_b}}{1-tR_{\mu_a\mu_b}}\prod_{i,a}\frac{1-D_{\lambda_i\mu_a}}{1-tD_{\lambda_i\mu_a}}\frac{1-t^2D_{\lambda_i\mu_a}}{1-tD_{\lambda_i\mu_a}}q_{[\lambda,\mu]}({\bf x},{\bf y})
\ee
where $1\leq i,j\leq l$ and $1\leq a,b\leq l'$.

For Young diagram $\lambda=(\lambda_1,\lambda_2,\cdots,\lambda_l)$, we use $|\lambda|$ to denote $\lambda_1+\lambda_2+\cdots+\lambda_l$. The notation $\lambda\geq\mu$ means $|\lambda|=|\mu|$ and $\lambda_1+\cdots+\lambda_i\geq\mu_1+\cdots+\mu_i$ for all $i$.
From (\ref{Qdef}), we have
\be\label{qqqqq}
Q_{[\lambda,\mu]}({\bf x},{\bf y})=\sum_{\nu\geq\lambda,\eta\geq\mu}a_{\nu\eta}^{\lambda\mu}(t)q_{[\nu,\eta]}({\bf x},{\bf y})+\sum_{|\xi|<|\lambda|,|\tau|<|\mu|}a_{\xi\tau}^{\lambda\mu}(t)q_{[\xi,\tau]}({\bf x},{\bf y})
\ee
where the polynomials $a_{\nu\eta}^{\lambda\mu}(t)\in\Z[t]$, $a_{\lambda\mu}^{\lambda\mu}(t)=1$. Then the polynomials $Q_{[\lambda,\mu]}({\bf x},{\bf y})$ form a basis of $\Q(t)[{\bf x},{\bf y}]$, i.e.,
\be\label{QQmult}
Q_{[\kappa,\theta]}Q_{[\nu,\eta]}=\sum M_{[\kappa,\theta],[\nu,\eta]}^{[\lambda,\mu]}Q_{[\lambda,\mu]}
\ee
where the structure constant $M_{[\kappa,\theta],[\nu,\eta]}^{[\lambda,\mu]}$ is given in the following.

Define the degree of each variables $x_n, y_n\  (n=1,2,\cdots)$ by
\[
\text{deg }x_n=n,\quad \text{deg }y_n=-n
\]
then $Q_{[\lambda,\mu]}({\bf x},{\bf y})$ is a homogeneous polynomial of degree $|\lambda|-|\mu|$.

Note that the Hall-Littlewood function $Q_\lambda({\bf x})$ is the special case $\mu=\emptyset$ of the coupled Hall-Littlewood function $Q_{[\lambda,\mu]}({\bf x},{\bf y})$:
\be
Q_\lambda({\bf x})=\prod_{ i<j}\frac{1-R_{ij}}{1-tR_{ij}}q_{\lambda}({\bf x})=Q_{[\lambda,\emptyset]}({\bf x},{\bf y}).
\ee
The generalized Q-functions defined in \cite{yuji} are the special case $t=-1$ of coupled Hall-Littlewood functions with a slight difference. The author finds that the generalized Q-functions give the tau functions of B-type UC hierarchy.
The universal character $S_{[\lambda,\mu]}({\bf x},{\bf y})$ is the special case $t=0$ of the coupled Hall-Littlewood function $Q_{[\lambda,\mu]}({\bf x},{\bf y})$:
\begin{eqnarray*}
S_{[\lambda,\mu]}({\bf x},{\bf y})&=&\text{det}\left( \begin{array}{cc}
h_{\mu_{l'-i+1}+i-j}({\bf y}), & 1\leq i\leq l' \\
h_{\lambda_{i-l'}-i+j}({\bf x}), & l'+1\leq i\leq l+l'
\end{array} \right)_{1\leq i,j\leq l+l'}\\
&=&\prod_{ i<j}(1-R_{\lambda_i\lambda_j})\prod_{a<b}(1-R_{\mu_a\mu_b})\prod_{i,a}(1-D_{\lambda_i\mu_a})h_{[\lambda,\mu]}({\bf x},{\bf y})
\end{eqnarray*}
where $h_{[\lambda,\mu]}({\bf x},{\bf y})=h_\lambda({\bf x})h_{\mu}({\bf y})$.

For example, $\lambda=(2,1),\ \mu=(1)$,
\begin{eqnarray*}
Q_{[\lambda,\mu]}({\bf x},{\bf y})&=&\frac{1-R_{\lambda_1\lambda_2}}{1-tR_{\lambda_1\lambda_2}}\frac{(1-D_{\lambda_1\mu_1})(1-t^2D_{\lambda_1\mu_1})}{(1-tD_{\lambda_1\mu_1})^2}\frac{(1-D_{\lambda_2\mu_1})(1-t^2D_{\lambda_2\mu_1})}{(1-tD_{\lambda_2\mu_1})^2}q_{[\lambda,\mu]}({\bf x},{\bf y})\\
&=&(1-(1-t)R_{\lambda_1\lambda_2})(1-(1-t)^2D_{\lambda_1\mu_1}-(1-t)^2D_{\lambda_2\mu_1})q_{[\lambda,\mu]}({\bf x},{\bf y})\\
&=&(q_2({\bf x})q_1({\bf x})-(1-t)q_3({\bf x}))q_1({\bf y})-(1-t)^2 q_1^2({\bf x})-t(1-t)^2q_2({\bf x}).
\end{eqnarray*}
In the special case $t=0$, $S_{[\lambda,\mu]}({\bf x},{\bf y})=(h_2({\bf x})h_1({\bf x})-h_3({\bf x}))h_1({\bf y})- h_1^2({\bf x})=(\frac{x_1^3}{3}-x_3)y_1-x_1^2$.

Introduce the following vertex operators
\bea
&&\Gamma_1^-(z)=e^{\xi_t({\bf x}-\tilde{\partial}_{\bf y}, z)},\quad \Gamma_1^+(z)=e^{\xi(\tilde{\partial}_{\bf x}, z^{-1})},\label{g1}\\
&&\Gamma_2^-(z)=e^{\xi_t({\bf y}-\tilde{\partial}_{\bf x}, z)},\quad \Gamma_2^+(z)=e^{\xi(\tilde{\partial}_{\bf y}, z^{-1})},\label{g2}
\eea
where $\xi({\bf x}, z)$ is the special case $t=0$ of $\xi_t({\bf x}, z)$.
Define
\bea
X^\pm(z)&=&\sum_{n\in\Z}X^\pm_nz^n=e^{\pm\xi_t({\bf x}-\tilde{\partial}_{\bf y}, z)}e^{\mp\xi(\tilde{\partial}_{\bf x}, z^{-1})},\\
Y^\pm(z)&=&\sum_{n\in\Z}Y^\pm_nz^{n}=e^{\pm\xi_t({\bf y}-\tilde{\partial}_{\bf x}, z)}e^{\mp\xi(\tilde{\partial}_{\bf y}, z^{-1})}.
\eea
The operators $X_i^\pm$ satisfy the following deformed fermionic relations:
\bea
X_{n-1}^\pm X_m^\pm-tX_{n}^\pm X_{m-1}^\pm +X_{m-1}^\pm X_n^\pm-tX_{m}^\pm X_{n-1}^\pm&=&0,\nonumber\\
X_n^+X_{m-1}^--t X_{n-1}^+X_{m}^-+X_m^+X_{n-1}^--t X_{m-1}^+X_{n}^-&=&(1-t)^2\delta_{m+n,1}.\nonumber
\eea
The same relations hold also for $Y_i^\pm$, and $X_i^\pm$ and $Y_i^\pm$ are commutative.
\begin{prp}
The operators $X_i^+$ and $Y_i^+$ are raising operators for the coupled Hall-Littlewood functions such that
\begin{eqnarray}
Q_{[\lambda,\mu]}({\bf x},{\bf y})&=&X_{\lambda_1}^+\cdots X_{\lambda_l}^+Y_{\mu_1}^+\cdots X_{\mu_{l'}}^+\cdot 1\\
&=&[z^\lambda w^\mu]X^+(z_1)\cdots X^+(z_l)Y^+(w_1)\cdots Y^+(w_{l'})\cdot 1,
\end{eqnarray}
where the Young diagrams $\lambda=(\lambda_1,\lambda_2,\cdots,\lambda_l),\ \mu=(\mu_1,\mu_2,\cdots,\mu_{l'})$ and the notation $[z^\lambda w^\mu]$ means taking the coefficient of $z_1^{\lambda_1}\cdots z_l^{\lambda_l}w_1^{\mu_1}\cdots w_{l'}^{\mu_{l'}}$ in the expansion of $X^+(z_1)\cdots X^+(z_l)Y^+(w_1)\cdots Y^+(w_{l'})\cdot 1$.
\end{prp}
\begin{proof}
We calculate
\begin{eqnarray*}
&& X^+(z_1)\cdots  X^+(z_l) Y^+(w_1)\cdots Y^+(w_{l'})\\
&=&e^{\xi_t({\bf x}-\tilde{\partial}_{\bf y}, z_1)}e^{-\xi(\tilde{\partial}_{\bf x}, z^{-1}_1)}\cdots e^{\xi_t({\bf x}-\tilde{\partial}_{\bf y}, z_l)}e^{-\xi(\tilde{\partial}_{\bf x}, z^{-1}_l)}\\
&&e^{\xi_t({\bf y}-\tilde{\partial}_{\bf x}, w_1)}e^{-\xi(\tilde{\partial}_{\bf y}, w^{-1}_1)}\cdots e^{\xi_t({\bf y}-\tilde{\partial}_{\bf x}, w_{l'})}e^{-\xi(\tilde{\partial}_{\bf y}, w^{-1}_{l'})}\cdot 1\\
&=&\prod_{ i<j}\frac{1-z_j/z_i}{1-tz_j/z_i}\prod_{a<b}\frac{1-w_b/w_a}{1-tw_b/w_a}e^{\xi_t({\bf x}-\tilde{\partial}_{\bf y}, z_1)}\cdots e^{\xi_t({\bf x}-\tilde{\partial}_{\bf y}, z_l)}e^{\xi_t({\bf y}-\tilde{\partial}_{\bf x}, w_1)}\cdots e^{\xi_t({\bf y}-\tilde{\partial}_{\bf x}, w_{l'})}\cdot 1\\
&=&\prod_{ i<j}\frac{1-z_j/z_i}{1-tz_j/z_i}\prod_{a<b}\frac{1-w_b/w_a}{1-tw_b/w_a}\prod_{i,a}\frac{1-z_iw_a}{1-tz_iw_a}\frac{1-t^2z_iw_a}{1-tz_iw_a}\\
&&e^{\xi_t({\bf x}, z_1)}\cdots e^{\xi_t({\bf x}, z_l)}e^{\xi_t({\bf y}, w_1)}\cdots e^{\xi_t({\bf y}, w_{l'})}\cdot 1.
\end{eqnarray*}
Taking the coefficient of $z_1^{\lambda_1}\cdots z_l^{\lambda_l}w_1^{\mu_1}\cdots w_{l'}^{\mu_{l'}}$, we finish the proof.
\end{proof}
Also we can  easily get the following two corollaries.
\begin{cor}
\label{qQQq} For Hall-Littlewood functions $Q_\lambda({\bf x})$ and $q_k({\bf x})$ defined in \cite{Mac}, we have
\be
q_k({\bf y}-\tilde{\partial}_{\bf x})Q_\lambda({\bf x}-\tilde{\partial}_{\bf y})=Q_\lambda({\bf x}-\tilde{\partial}_{\bf y})q_k({\bf y}-\tilde{\partial}_{\bf x})
\ee
\end{cor}

\begin{cor}
The coupled Hall-Littlewood function can be written as
\be\label{QQQ}
Q_{[\lambda,\mu]}({\bf x},{\bf y})=Q_{\lambda}({\bf x}-\tilde{\partial}_{\bf y})Q_{\mu}({\bf y}-\tilde{\partial}_{\bf x})\cdot 1
\ee
where $Q_{\lambda}({\bf x})$ is the Hall-Littlewood function.
\end{cor}
By a direct calculation, the following proposition holds true.
\begin{prp}
The formula (\ref{QQmult}) can be written in the following algebraic form
\begin{eqnarray}
Q_{[\kappa,\theta]}Q_{[\nu,\eta]}&=&\prod_{ij}\frac{(1-tD_{\kappa_i\eta_j})^2}{(1-D_{\kappa_i\eta_j})(1-t^2D_{\kappa_i\eta_j})}\prod_{mn}\frac{(1-tD_{\theta_m\nu_n})^2}{(1-D_{\theta_m\nu_n})(1-t^2D_{\theta_m\nu_n})}\\
&&(Q_\kappa({\bf x}-\tilde{\partial}_{\bf y})Q_\nu({\bf x}-\tilde{\partial}_{\bf y})Q_{\theta}({\bf y}-\tilde{\partial}_{\bf x})Q_{\eta}({\bf y}-\tilde{\partial}_{\bf x}))\cdot 1
\end{eqnarray}
where
\be
(Q_\kappa({\bf x}-\tilde{\partial}_{\bf y})Q_\nu({\bf x}-\tilde{\partial}_{\bf y})Q_{\theta}({\bf y}-\tilde{\partial}_{\bf x})Q_{\eta}({\bf y}-\tilde{\partial}_{\bf x}))\cdot 1=\sum_{\lambda\mu}f_{\kappa\nu}^\lambda f_{\theta\eta}^\mu Q_{[\lambda,\mu]}
\ee
and $f_{\kappa\nu}^\lambda$, which is defined in \cite{Mac}, is the coefficient of Hall-Littlewood function $Q_{\lambda}({\bf x})$ in the expansion of $Q_\kappa({\bf x})Q_\nu({\bf x})$.
\end{prp}
As in \cite{Mac}, we use $\lambda'$ to denote the transpose of the Young diagram $\lambda$ and $\theta=\lambda-\mu$ a skew diagram. If $|\theta|=n$ and $\theta'_i=1$ for each $i\geq 1$, we say $\theta$ is a horizontal $n$-strip.
Let
\be\label{psi11}
\psi_{\lambda/\mu}(t)=\prod_{j\in J_{\lambda/\mu}}(1-t^{m_j(\mu)})
\ee
where $\theta=\lambda-\mu$ is a horizontal strip, $m_i(\mu)$ denotes the number of times $i$ occurs as a part of $\mu$, and $J_{\lambda/\mu}$ is the set of integers $j\geq 1$ such that $\theta'_j<\theta'_{j+1}$.
We have
\be\label{qQ}
q_n({\bf x}) Q_\mu({\bf x})=\sum_\lambda \psi_{\lambda/\mu}(t) Q_\lambda({\bf x})
\ee
summed over all $\lambda\supset\mu$ such that $\lambda-\mu$ is a horizontal $n$-strip.

For a pair of Young diagrams $\lambda$ and $\mu$, the notation $\lambda\succ\mu$ or $\mu\prec\lambda$ means that the Young diagrams $\lambda$ and $\mu$ are interlaced, in the sense of $\lambda_1\geq\mu_1\geq\lambda_2\geq\mu_2\geq\cdots$. Let the skew Hall-Littlewood function $P_{\lambda/\mu}$ of a single variable $z$, indexed by the skew Young diagram $\lambda/\mu$, is
\be
P_{\lambda/\mu}(z,t)=\psi_{\lambda/\mu}(t)z^{|\lambda|-|\mu|},\quad \lambda\succ\mu
\ee
and $P_{\lambda/\mu}(z,t)=0$ otherwise.

\begin{prp}
The vertex operators in (\ref{g1}) and (\ref{g2}) act on the coupled Hall-Littlewood functions $Q_{[\lambda,\mu]}({\bf x},{\bf y})$ in the following way,
\bea
\Gamma_1^-(z)Q_{[\lambda,\mu]}({\bf x},{\bf y})&=&\sum_{\nu\succ\lambda}P_{\nu/\lambda} (z,t) Q_{[\nu,\mu]}({\bf x},{\bf y}),\\
\Gamma_2^-(z)Q_{[\lambda,\mu]}({\bf x},{\bf y})&=&\sum_{\nu\succ\mu} P_{\nu/\mu} (z,t)Q_{[\lambda,\nu]}({\bf x},{\bf y}),
\eea
where the sums are over all Young diagrams $\nu$.
\end{prp}

\section{$q$-boson model and Hall-Littlewood functions}\label{sect3}
In this section, we consider the $q$-boson model based on the $q$-boson algebra\cite{PS,NVT}. The $q$-boson algebra is generated by three operators $B, B^\dag, N$ with commutation relations
\begin{equation}\label{nphi}
[N,B] = - B,\quad [N,B^{\dag}] = B^{\dag},\quad [B,B^{\dag}] = q^N.
\end{equation}
The representation of this algebra in the one-dimensional Fock space $\mathcal{F}$ consisting of $n$-particle states
$|n\rangle$ is as follows:
$$
B^{\dag}|n\rangle =|n+1\rangle,\quad B|n\rangle = [n] |n-1\rangle, \quad B|0\rangle = 0, \quad
N|n\rangle = n|n\rangle,
$$
where $|0\rangle$ is the vacuum state, the special case $n=0$ of the $n$ particle state, and we denote
\begin{equation}
[n]=\frac{1-q^n}{1-q},\quad [n]!=\prod_{j=1}^n[j].
\end{equation}
The scalar product is given by
\begin{equation}\label{scann}
\langle n|n\rangle=[n]!.
\end{equation}

When $q=1$, the $q$-boson operators become ordinary bosons $B\rightarrow b,\ B^\dag\rightarrow b^\dag$ which satisfy $[b,b^\dag]=1$.

Let the tensor product
\be
\mathcal{F} = \bigotimes_{i=0}^M \mathcal{F}_i
\ee
be $M+1$ copies of the Fock space. Denote by $B_i, B^{\dag}_i, N_i$ the operators that act as $B, B^{\dag}, N$ in (\ref{nphi}), respectively, in the $i$th space and identically in the other spaces. We call the $j$th Fock space $\mathcal{F}_j$ the $j$-energy space. Define the operator of the total number of particles by
$$
\hat{N} = \sum_{i=0}^M  N_i.
$$
Then the $N$-particle vectors in this space are of the form
\be
\bigotimes_{i=0}^M |n_i\rangle_i,\qquad \textrm{where} \ |n_i\rangle_i =
(B_i^{\dag})^{n_i}|0\rangle_i,\quad  N=\sum_{i=0}^M n_i,    \label{Npar}
\ee
which is denoted by $|\{n_i\}_{i=1}^M\rangle$. Using the equation (\ref{scann}),
we get
\be
\langle\{n_i\}_{i=1}^M|\{n_i\}_{i=1}^M\rangle=\prod_{j=0}^M[n_j]!
\ee

The $q$-boson model is a model of a chain with the hamiltonian
\be
H = -\frac{1}{2} \sum_{i=0}^M \big(B^{\dag}_i B_{i+1} + B_i B^{\dag}_{i+1} -2N_i \big),
\ee
for the periodic boundary conditions: $M+1\equiv 1$.
Note that the limit at $q=0$  of $q$-boson model is the phase model.

Introduce the $L$-matrix for the $q$-boson model
$$
L_i(u) = \left( \begin{array}{cc}
u^{-1} & B^{\dag}_i \\
(1-q)B_i & u
\end{array} \right), \qquad i=0,\ldots, M,
$$
where $u$ is a scalar parameter, here we treat $u$ as $uI$ with $I$ being the identity operator in $\mathcal F$. For every $i=0,\ldots, M$, the $L$-matrix satisfies the bilinear equation
\be\label{rll}
R(u,v)\big(L_i(u)\otimes L_i(v)  \big)  =  \big(L_i(v)\otimes L_i(u)  \big) R(u,v),
\ee
where $R$-matrix $R(u,v)$ is a $4\times4$ matrix given by
\be
R(u,v) = \left( \begin{array}{cccc}
f(v,u) & 0 & 0 & 0 \\
0 & g(v,u) & q^{-1/2} & 0 \\
0 & q^{1/2} & g(v,u) & 0 \\
0 & 0 & 0 & f(v,u)
\end{array} \right),   \label{Rmatrix}
\ee
with
$$
f(v,u) = \frac{q^{-1/2}u^2-q^{1/2}v^2}{u^2-v^2},\quad g(v,u) = \frac{uv( q^{-1/2}- q^{1/2})}{u^2-v^2}.
$$

Define the monodromy matrix by
$$
T(u) = L_M(u) L_{M-1}(u)\cdots L_0(u)
$$
which gives the solution of the $q$-boson model. It also satisfies the bilinear equation
\be\label{rtt}
R(u,v)\big(T(u)\otimes T(v)  \big)  = \big(T(v)\otimes T(u)  \big) R(u,v).
\ee

The objects we consider in the following are entriees of the monodromy matrix, which are denoted by
\be
T(u) =  \left( \begin{array}{cc}
A(u) & B(u) \\
C(u) & D(u)
\end{array} \right),    \label{ABCD}
\ee
the most important relations of (\ref{rtt}) are
\begin{eqnarray}
q^{-1/2}A(u)B(v)&=&f(u,v)B(v)A(u)+g(v,u)B(u)A(v),\\
q^{-1/2}D(u)B(v)&=&f(v,u)B(v)D(u)+g(u,v)B(u)D(v),\label{ytt2}\\
C(u)B(v)-qB(v)C(u)&=&q^{1/2}g(u,v)(A(u)D(v)-A(v)D(u)).\label{ytt3}
\end{eqnarray}

Since
\be
\hat{N} B(u) = B(u) (\hat{N}+1), \qquad \hat{N} C(u) = C(u) (\hat{N}-1).   \label{creat-annih}
\ee
we call $B(u)$ the creation operator and $C(u)$ the annihilation operator. The  operators $A(u)$ and $D(u)$ do not change the total number of particles.

Denote by $|0\rangle_j$ the vacuum vector in $\mathcal F_j$ and by $|0\rangle=\otimes_{i=0}^M |0\rangle_i$. Let
\be
|\Psi(u_1,\ldots,u_N) \rangle = \prod_{i=1}^N B(u_j) |0\rangle,   \label{Psi}
\ee
which is a $N$ particle state.

According to \cite{PS,NVT}, there is the following isometry between the states (\ref{Npar}) and the Hall-Littlewood functions
\begin{equation}\label{map}
\bigotimes_{i=0}^{M} |n_i\rangle_i\mapsto |\lambda\rangle=Q_{\lambda}({\bf x}),\quad \lambda=1^{n_1}2^{n_2}\ldots
\end{equation}
and the operator $B(u)$ acts on $Q_{\lambda}({\bf x})$ as the operator of multiplication by $u^{-M}H_M(u^2)$, where $H_M(t)=\sum_{k=0}^Mt^kq_k({\bf x})$ is the truncated generating function of $q_k({\bf x})$ defined in (\ref{qr}). Then the state $|\Psi(u_1,\ldots,u_N) \rangle$ has the following expansion
\begin{equation}
|\Psi(u_1,\ldots,u_N) \rangle=\sum_\lambda P_\lambda(u_1^2,\cdots,u_N^2)|\{n_i\}_{i=1}^M\rangle,\quad \lambda=1^{n_1}2^{n_2}\cdots
\end{equation}
where $P_\lambda({\bf x})$ is the Hall-Littlewood function defined in \cite{Mac}.

By forgetting $B_0^+$ and $(1-q)B_0$, we find that the operator $C(u)$ acts on $Q_{\lambda}({\bf x})$ as the operator $u^{-M}H_M^\bot(u^2)$, where $H_M^\bot(t)=\sum_{k=0}^Mt^kq_k^\bot({\bf x})$ and $q_k^\bot({\bf x})$ is adjoint to $q_k({\bf x})$ with respect to scalar product $(P_\lambda({\bf x}),Q_\mu({\bf x}))=\delta_{\lambda,\mu}$ in the space $\Lambda_M$ generated by Hall-Littlewood functions $Q_\lambda({\bf x})$ whose diagrams have at most $M$ columns. Then the state $\langle\Psi(u_1,\ldots,u_N) |$ has the following expansion
\begin{equation}
\langle\Psi(u_1,\ldots,u_N) |=\sum_\lambda\langle\{n_i\}_{i=1}^M| P_\lambda(u_1^{-2},\cdots,u_N^{-2}),\quad \lambda=1^{n_1}2^{n_2}\cdots.
\end{equation}

In the following, we will construct the realization of the two-site generalized $q$-boson model in the algebra of coupled Hall-Littlewood functions.
\section{Two-site generalized $q$-boson model and coupled Hall-Littlewood functions}\label{sect4}
In this section, let the coupled Hall-Littlewood functions be in $\C(t)[{\bf x},{\bf y}]$. For two positive integers $M_1,M_2$, we consider the $M_1+M_2+2$ copies of the Fock space of $q$-boson algebra \be\mathcal{F}=\mathcal{F}^{(1)}\otimes \mathcal{F}^{(2)},\ee where the
 tensor products
\be
\mathcal{F}^{(1)} = \bigotimes_{i=0}^M \mathcal{F}_i^{(1)}\quad\quad \mathcal{F}^{(2)} = \bigotimes_{i=0}^M \mathcal{F}_i^{(2)},
\ee
 are $M_1+1$ and $M_2+1$ copies of the Fock space of $q$-boson algebra respectively. We use $B_i^{(1)}, B^{(1)\dag}_i, N_i^{(1)}$ to denote the operators that act as $B, B^{\dag}, N$ in (\ref{nphi}), respectively, in the $i$th space $\mathcal{F}_i^{(1)}$ and identically in the other spaces of $\mathcal{F}^{(1)}$ and in all spaces of $\mathcal{F}^{(2)}$, and $B_i^{(2)}, B^{(2)\dag}_i, N_i^{(2)}$ to denote the operators that act as $B, B^{\dag}, N$, respectively, in the $i$th space $\mathcal{F}_i^{(2)}$ and identically in the other spaces of $\mathcal{F}^{(2)}$ and in all spaces of $\mathcal{F}^{(1)}$.
 We denote  the vacuum vector in $\mathcal F_j^{(i)}$ by $|0\rangle_j^{(i)}$ for $i=1,2.$

The operator of the total number of particles is given by
\be
\hat N=\hat{N}_1+\hat{N}_2
\ee
where the operators
$$
\hat{N}_1 = \sum_{i=0}^M  N_i^{(1)}, \quad \hat{N}_2 = \sum_{i=0}^M  N_i^{(2)}.
$$
Then the $(N_1,N_2)$-particle vectors in space $\mathcal{F}$ are of the form
\be\label{nmbasis}
\bigotimes_{i=0}^{M_1} |n_i\rangle_i^{(1)}\bigotimes\bigotimes_{i=0}^{M_2} |m_i\rangle_i^{(2)},\qquad \textrm{with}\quad N_1=\sum_{i=0}^{M_1} n_i,\quad N_2=\sum_{i=0}^{M_2} m_i,
\ee
where
\[
|n_i\rangle_i^{(1)}= (B_i^{(1)\dag})^{(n_i)}|0\rangle_j^{(1)},\quad |m_i\rangle_i^{(2)}= (B_i^{(2)\dag})^{(m_i)}|0\rangle_j^{(2)}.
\]

Define the map
$\jmath: \mathcal F\rightarrow \C(t)[{\bf x}, {\bf y}]$
by
\be\label{jmath}
\jmath(\bigotimes_{i=0}^{M_1} |n_i\rangle_i^{(1)}\bigotimes\bigotimes_{i=0}^{M_2} |m_i\rangle_i^{(2)})= Q_{[\lambda,\mu]}({\bf x},{\bf y})
\ee
with
\be
\lambda=1^{n_1}2^{n_2}\ldots,\quad \mu=1^{m_1}2^{m_2}\ldots.
\ee
Note that partitions $\lambda,\ \mu$ themselves have no relations with numbers $n_0$ and $m_0$ of particles. Nonetheless, if we fix the total numbers of particles $N_1$ and $N_2$, we can deduce $n_0=N_1-l(\lambda)$ and $m_0=N_2-l(\mu)$, where $l(\lambda)$ is the length of partition $\lambda$, that is, the number of rows in $\lambda$.
Therefore the map (\ref{jmath}) gives a representation of
\be
\hat {\mathcal F}=\bigotimes_{i=1}^M \mathcal{F}_i^{(1)}\bigotimes \bigotimes_{i=1}^M \mathcal{F}_i^{(2)}.
\ee
in the algebra of $\C(t)[{\bf x},{\bf y}]$, in fact in its subspace $\C(t)_{M_1,M_2}[{\bf x},{\bf y}]$ generated by coupled Hall-Littlewood functions $Q_{[\lambda,\mu]}({\bf x},{\bf y})$ where the Young diagrams $\lambda$ have at most $M_1$ columns and $\mu$ at most $M_2$ columns.
The correspondence (34) in \cite{PS} is a special case of the map $\jmath$ defined above.

The monodromy matrix is given by
\bea
T(u)&=&L^{(2)}_{M_2}(u) \cdots L^{(2)}_0(u)L^{(1)}_{M_1}(u) \cdots L^{(1)}_0(u),\nonumber
\eea
where the $L$-matrices
$$
L^{(1)}_i(u) = \left( \begin{array}{cc}
u^{-1} & B^{(1)\dag}_i \\
(1-q)B_i^{(1)} & u
\end{array} \right), \qquad i=0,\ldots, M_1,
$$
$$
L^{(2)}_i(u) = \left( \begin{array}{cc}
u^{-1} & B^{(2)\dag}_i \\
(1-q)B_i^{(2)} & u
\end{array} \right), \qquad i=0,\ldots, M_2.
$$
From the relation (\ref{rll}), we get that each $L$-matrix and the monodromy matrix satisfy the bilinear equation again
\bea
R(u,v)\big(L^{(j)}_i(u)\otimes L^{(j)}_i(v)  \big) & = & \big(L^{(j)}_i(v)\otimes L^{(j)}_i(u)  \big) R(u,v),\quad j=1,2,\nonumber \\
R(u,v)\big(T(u)\otimes T(v)  \big) & = & \big(T(v)\otimes T(u)  \big) R(u,v), \label{intertwine}
\eea
with the same $R$-matrix in (\ref{Rmatrix}).

Let
\bea
T_i(u)&=&L^{(i)}_{M_i}(u)L^{(i)}_{M_i-1}(u) \cdots L^{(i)}_0(u)\\
&=&\left( \begin{array}{cc}
A_i(u) & B_i(u) \\
C_i(u) & D_i(u)
\end{array} \right)\quad i=1,2.
\eea
We call the operators $B_1(u)$ and $B_2(u)$ the creation operators and $C_1(u)$ and $C_2(u)$ the annihilation operators since
\be
\hat{N}_i B_i(u) = B_i(u) (\hat{N}_i+1), \qquad \hat{N}_i C_i(u) = C_i(u) (\hat{N}_i-1) \quad\text{for}\quad i=1,2.
\ee
The operators $A_i(u)$ and $D_i(u)\ (i=1,2)$   do not change the total number of particles.

Consider the decomposition
\be
\mathcal F=\bigoplus_{N_1,N_2\geq 0}\mathcal F^{N_1,N_2}=\bigoplus_{N_1,N_2\geq 0}\mathcal F_1^{N_1}\otimes\mathcal F_2^{N_2}
\ee
of the whole space $\mathcal {F}$ into $(N_1,N_2)$-particle subspaces $\mathcal{F}^{N_1,N_2}$.
Let $\C(t)_{M_1,M_2}^{N_1,N_2}[{\bf x},{\bf y}]$ be the space of coupled Hall-Littlewood functions corresponding to $\mathcal F^{N_1,N_2}$. We find that $\C(t)_{M_1,M_2}^{N_1,N_2}[{\bf x},{\bf y}]$ is spanned by coupled Hall-Littlewood functions $Q_{[\lambda,\mu]}({\bf x},{\bf y})$ whose diagrams $\lambda$ lie in the $N_1\times M_1$ box and $\mu$ lie in the $N_2\times M_2$ box. That Young diagram $\lambda$ lies in the $N\times M$ box means $\lambda$ has at most $N$ rows and at most $M$ columns.

 We know that $B_i(u)$ are creation operators. In the following, we discuss the actions of ${B}_i(u)$ on $\hat {\mathcal F}=\C(t)_{M_1,M_2}[{\bf x},{\bf y}]$. Define the projection $P:\mathcal F\rightarrow \hat{\mathcal F}$ by forgetting the zero energy states. Thus the operators we considered are $PB_i(u)P$, we still denote it by $B_i(u)$. Then  ${B}_1(u)$ sends $\C(t)_{M_1,M_2}^{N_1,N_2}[{\bf x},{\bf y}]$ to $\C(t)_{M_1,M_2}^{N_1+1,N_2}[{\bf x},{\bf y}]$ and $B_2(u)$ sends $\C(t)_{M_1,M_2}^{N_1,N_2}[{\bf x},{\bf y}]$ to $\C(t)_{M_1,M_2}^{N_1,N_2+1}[{\bf x},{\bf y}]$.

Define $\tilde{B}_i(u)$ by $B_i(u)=u^{-M_i}\tilde{B}_i(u)$ for $i=1,2$.

\begin{prp}\label{prpbb}
In the space $\C(t)_{M_1,M_2}[{\bf x},{\bf y}]$,
\bea
\tilde{B}_1(u)&=& H_{M_1}({\bf x}-\tilde{\partial}_{\bf y},u^2),\\
\tilde{B}_2(u)&=& H_{M_2}({\bf y}-\tilde{\partial}_{\bf x},u^2),
\eea
where $H_n({\bf x},z)=\sum_{k=0}^n z^kq_k({\bf x})$, and $q_k({\bf x})$ is defined in (\ref{qr}).
\end{prp}
\begin{proof}
From the results in \cite{PS,NVT} and the equation (\ref{QQQ}), we have
\bea
&&\tilde{B}_1(u)\bigotimes_{i=0}^{M_1} |n_i\rangle_i^{(1)}\bigotimes\bigotimes_{i=0}^{M_2} |m_i\rangle_i^{(2)}\\
&=&\big(\tilde{B}_1(u)\bigotimes_{i=0}^{M_1} |n_i\rangle_i^{(1)}\big)\bigotimes\bigotimes_{i=0}^{M_2} |m_i\rangle_i^{(2)}\nonumber\\
&=& \sum_{k=0}^{M_1}u^{2k}q_k({\bf x}-\tilde{\partial}_{\bf y}) Q_\lambda({\bf x}-\tilde{\partial}_{\bf y})Q_\mu({\bf y}-\tilde{\partial}_{\bf x})\cdot 1,\nonumber
\eea
and
\bea
&&\tilde{B}_2(u)\bigotimes_{i=0}^{M_1} |n_i\rangle_i^{(1)}\bigotimes\bigotimes_{i=0}^{M_2} |m_i\rangle_i^{(2)}\\
&=&\bigotimes_{i=0}^{M_1} |n_i\rangle_i^{(1)}\bigotimes\big(\tilde{B}_2(u)\bigotimes_{i=0}^{M_2} |m_i\rangle_i^{(2)}\big)\nonumber\\
&=&Q_\lambda({\bf x}-\tilde{\partial}_{\bf y})\big(\sum_{k=0}^{M_2}u^{2k}q_k({\bf y}-\tilde{\partial}_{\bf x})Q_\mu({\bf y}-\tilde{\partial}_{\bf x})\big)\cdot 1\nonumber\\
&=& \sum_{k=0}^{M_2}u^{2k}q_k({\bf y}-\tilde{\partial}_{\bf x}) Q_\lambda({\bf x}-\tilde{\partial}_{\bf y})Q_\mu({\bf y}-\tilde{\partial}_{\bf x})\cdot 1,\nonumber
\eea
here the mapping signs $\jmath$ are omitted on the left hand sides of the equations.
\end{proof}

By forgetting $B_0^+$ and $(1-q)B_0$, we obtain the following formulas:
\begin{lem}
The operators ${A}_i(u),\ B_i(u),\ {C}_i(u),\ {D}_i(u),\ (i=1,2)$ are related by
\[
{A}_i(u)=u^{-1}B_i(u),\ {C}_i(u)=B^{\dag}_i(u^{-1}),\ {D}_i(u)=uB^{\dag}_i(u^{-1}).
\]
\end{lem}
Since $B_1(u)=u^{-M_1}H_{M_1}({\bf x}-\tilde{ \partial}_{\bf y}, u^2)$ and $B_2(u)=u^{-M_2}H_{M_2}({\bf y}-\tilde{ \partial}_{\bf x}, u^2)$, we obtain the following lemma.
\begin{lem}\label{ccc} Let $\tilde{{C}}_i(u)=u^{-M_i}{C}_i(u)$ for $i=1,2$, we have
\[
\tilde{{C}}_1(u)= H_{M_1}^\bot({\bf x}-\tilde{\partial}_{\bf y},u^{-2}),
\]
\[
\tilde{{C}}_2(u)= H_{M_2}^\bot({\bf y}-\tilde{\partial}_{\bf x},u^{-2}),
\]
where $H_{M}^\bot({\bf x},t)=\sum_{k=0}^{M}t^kq_k^{\bot}({\bf x})$, and $q_k^{\bot}({\bf x})$ is adjoint to the operator of multiplication by $q_k({\bf x})$.
\end{lem}
Note that the operators $q_k({\bf x})$ and $q_k^{\bot}({\bf x})$ are adjoint to each other in the lemma above with respect to the scalar product $(P_\lambda({\bf x}),Q_\lambda({\bf x}))=\delta_{\lambda\mu}$ in the space $\Lambda_{M}$ spanned by the Hall-Littlewood function $Q_\lambda({\bf x})$ whose diagrams $\lambda$ have at most $M$ columns.
\begin{cor}
The map $\bot:q_k({\bf x})\rightarrow q_k^{\bot}({\bf x})$ in the algebra of symmetric functions corresponds to the map $+:B^+\rightarrow (1-q)B, \ B\rightarrow\frac{1}{1-q}B^+$ in $q$-boson model.
\end{cor}

From proposition (\ref{prpbb}) and lemma (\ref{ccc}), we obtain the following proposition.
\begin{prp}
In the $M_1,M_2\rightarrow \infty$ limit, operators $\tilde{B}_i(u)$ and $\tilde{\mathcal C}_i(u)$ have the following vertex operator representations
\bea
\tilde{B}_1(u)&=&e^{\xi_t({\bf x}-\tilde{\partial}_{\bf y},u^2)}=\Gamma_1^-(u^2),\\
\tilde{\mathcal C}_1(u)&=&e^{\xi(\tilde{\partial}_{\bf x},u^{-2})}\ =\Gamma_1^+(u^2),\\
\tilde{B}_2(u)&=&e^{\xi_t({\bf y}-\tilde{\partial}_{\bf x},u^2)}=\Gamma_2^-(u^2),\\
\tilde{\mathcal C}_2(u)&=&e^{\xi(\tilde{\partial}_{\bf y},u^{-2})}\  =\Gamma_2^+(u^2),
\eea
where the vertex operators $\Gamma_i^\pm(t),i=1,2$ are defined in (\ref{g1}) and (\ref{g2}).
\end{prp}
Define
\be\label{tildepsi}
|\tilde\Psi_N(u_1,\cdots,u_N)\rangle:=\prod_{j=1}^NB_2(u_j)B_1(u_j)|0\rangle
\ee
then we obtain the following proposition which gives the expansion of the $(N,N)$-particle vector (\ref{tildepsi}).
\begin{prp}
The expansion of the $(N,N)$-particle vector (\ref{tildepsi}) in terms of basis vector (\ref{nmbasis}) is given by the formula
\bea
&&|\tilde\Psi_N(u_1,\cdots,u_N)\rangle\nonumber\\
&=&(u_1\cdots u_N)^{-M_1-M_2}\sum_{\lambda,\mu}P_\lambda(u_1^2,\cdots,u_N^2)P_\mu(u_1^2,\cdots,u_N^2)\bigotimes_{i=0}^{M_1} |n_i\rangle_i^{(1)}\bigotimes\bigotimes_{i=0}^{M_2} |m_i\rangle_i^{(2)}\nonumber\\
&=&(u_1\cdots u_N)^{-M_1-M_2}\sum_{\lambda,\mu}P_\lambda(u_1^2,\cdots,u_N^2)P_\mu(u_1^2,\cdots,u_N^2)Q_{[\lambda,\mu]}({\bf x},{\bf y}),\nonumber
\eea
where the sum is over Young diagrams $\lambda$ with at most $N$ rows and at most $M_1$ columns, Young diagrams $\mu$ with at most $N$ rows and at most $M_2$ columns.
\end{prp}
\begin{proof}
The operators $B_1(u)$ and $B_2(u)$ are commutative and
\bea
\prod_{j=1}^NB_1(u_j)&=&(u_1\cdots u_N)^{-M_1}\sum_{\lambda}P_{\lambda}(u_1^2,\cdots,u_N^2)Q_\lambda({\bf x}-\tilde{ \partial}_{\bf y}),\\
\prod_{j=1}^NB_2(u_j)&=&(u_1\cdots u_N)^{-M_2}\sum_{\mu}P_{\mu}(u_1^2,\cdots,u_N^2)Q_\mu({\bf y}-\tilde{ \partial}_{\bf x}),
\eea
then we have
\bea
&&|\tilde\Psi_N(u_1,\cdots,u_N)\rangle\nonumber\\
&=&(u_1\cdots u_N)^{-M_1-M_2}\sum_{\lambda,\mu}P_\lambda(u_1^2,\cdots,u_N^2)P_\mu(u_1^2,\cdots,u_N^2)Q_\lambda({\bf x}-\tilde{ \partial}_{\bf y})Q_\mu({\bf y}-\tilde{ \partial}_{\bf x})\cdot1\nonumber\\
&=& (u_1\cdots u_N)^{-M_1-M_2}\sum_{\lambda,\mu}P_\lambda(u_1^2,\cdots,u_N^2)P_\mu(u_1^2,\cdots,u_N^2)Q_{[\lambda,\mu]}({\bf x},{\bf y}).
\eea
By the restrictions on $\lambda$ and $\mu$, we obtain the conclusion.
\end{proof}

Recall that
\bea
&&T(u)=\left( \begin{array}{cc}
A(u) & B(u) \\
C(u) & D(u)
\end{array} \right)\nonumber\\
&=&T_2(u)T_1(u)=\left( \begin{array}{cc}
A_2(u) & B_2(u) \\
C_2(u) & D_2(u)
\end{array} \right)\left( \begin{array}{cc}
A_1(u) & B_1(u) \\
C_1(u) & D_1(u)
\end{array} \right),\nonumber
\eea
then
\bea
B(u)&=&A_2(u)B_1(u)+ B_2(u) D_1(u)\\
&=&B_2(u)(u^{-1}B_1(u)+uB_1^\dag(u^{-1})).\nonumber
\eea
Then we need the following formulas:
\begin{prp}
From the bilinear equation (\ref{rtt}), we have
\be
B_i^\dag(u^{-1})B_i(v)=\frac{u^2-tv^2}{u^2-v^2}B_i(v)B_i^\dag(u^{-1})-\frac{v^2(1-t)}{u^2-v^2}B_i(u)B_i^\dag(v^{-1}).
\ee
for $i=1,2$.
\end{prp}
Further we can derive the following corollary.
\begin{cor}
The following equation holds
\be
H_M^\bot(z)H_M(w)=\frac{1-tzw}{1-zw}H_M(w)H_M^\bot(z)-\frac{1-t}{1-zw}(zw)^{M+1}H_M(z^{-1})H_M^\bot(w^{-1}),
\ee
where $H_M(z)$ denotes $H_M({\bf x}, z)$.
\end{cor}
Recall that \[
|\Psi_N(u_1,\cdots,u_N)\rangle=\prod_{j=1}^N B(u_j)|0\rangle.
\]
 Then by a direct calculation, we get the following proposition.
\begin{prp}
The operators $u_1^{-1}B_1(u_1)+u_1B_1^\dag(u_1^{-1})$ and $u_2^{-1}B_1(u_2)+u_2B_1^\dag(u_2^{-1})$ are commutative, which tells us that the coefficients, in the expansion $|\Psi_N(u_1,\cdots,u_N)\rangle$ in terms of basis vector (\ref{nmbasis}), are symmetric functions of variables $u_1^2,\cdots,u_N^2$.
\end{prp}
Since \bea
\prod_{j=1}^N B(u_j)|0\rangle&=&\prod_{j=1}^N B_2(u_j)(u_j^{-1}B_1(u_j)+u_jB_1^\dag(u_j^{-1}))|0\rangle\nonumber\\
&=&\prod_{j=1}^N B_2(u_j)\prod_{j=1}^N(u_j^{-1}B_1(u_j)+u_jB_1^\dag(u_j^{-1}))|0\rangle,\nonumber
\eea
and
\[
\prod_{j=1}^N B_2(u_j)=(u_1\cdots u_N)^{-M_2}\sum_{\mu}P_\mu(u_1^2,\cdots,u_N^2)Q_\mu({\bf y}-\tilde{ \partial}_{\bf x}),
\]
where $\mu$ is a Young diagram with at most $N$ rows and at most $M_2$ columns.
Hence, in the following, we consider the expansion of $\prod_{j=1}^N(u_j^{-1}B_1(u_j)+u_jB_1^\dag(u_j^{-1}))|0\rangle$ by by inductions in the following proposition.

\begin{prp}
Let $k_1,k_2,\cdots,k_i$ be in the set $\{1,2,\cdots, N\}$ and satisfy $k_1<k_2<\cdots<k_i$. We denote $u_{k_1}u_{k_2}\cdots u_{k_i}$ by $u_{\{k\}}$ for short. Then we have
\bea
&&\prod_{j=1}^N(u_j^{-1}B_1(u_j)+u_jB_1^\dag(u_j^{-1}))|0\rangle\nonumber\\
&=&(u_1\cdots u_N)^{M_1+1}\sum_{i=0}^N\sum_{k_1,\cdots,k_i}(u_{\{k\}})^{-2M_1-2}\prod_{j\neq k_i}\frac{u_j^2-tu_{k_i}^2}{u_j^2-u_{k_i}^2}H_1(u_{k_1}^2)\cdots H_1(u_{k_i}^2)\cdot 1,\nonumber
\eea
where $H_1(u^2)=H_{M_1}({\bf x}-\tilde{\partial}_{\bf y},u^2)$ and
\[
H_1(u_{k_1}^2)\cdots H_1(u_{k_i}^2)\cdot 1=\sum_{\lambda}P_{\lambda}(u_{k_1}^2,\cdots,u_{k_i}^2)Q_{\lambda}({\bf x}-\tilde{\partial}_{\bf y})\cdot 1.
\]
\end{prp}
From the discussion above, we get the expansion of $|\Psi_N(u_1,\cdots,u_N)\rangle$ as following.
\begin{prp}The $N$-particle vector $|\Psi_N(u_1,\cdots,u_N)\rangle$ can be written as
\bea
|\Psi_N(u_1,\cdots,u_N)\rangle&=&(u_1\cdots u_N)^{-M_2+M_1+1}\sum_{i=0}^N\sum_{k_1,\cdots,k_i}(u_{\{k\}})^{-2M_1-2}\prod_{j\neq k_i}\frac{u_j^2-tu_{k_i}^2}{u_j^2-u_{k_i}^2}\nonumber\\
&&\sum_{\lambda,\mu}P_{\lambda}(u_{k_1}^2,\cdots,u_{k_i}^2)P_\mu(u_1^2,\cdots,u_N^2)Q_{[\lambda,\mu]}({\bf x},{\bf y}),\nonumber
\eea
where $\lambda$ is a Young diagram with at most $i$ rows and $M_1$ columns, and $\mu$ a Young diagram with at most $N$ rows and $M_2$ columns.
\end{prp}
In special case $\mu=\emptyset$, we have $Q_{[\lambda,\emptyset]}({\bf x},{\bf y})=Q_{\lambda}({\bf x})$. Let $M_2=\emptyset$,
\bea
|\Psi_N(u_1,\cdots,u_N)\rangle&=&\prod_{j=1}^N B_1(u_j)|0\rangle=(u_1\cdots u_N)^{-M_1}\prod_{j=1}^N \tilde{B}_1(u_j)|0\rangle\nonumber\\
&=&(u_1\cdots u_N)^{-M_1}\sum_{\lambda}P_\lambda(u_1^2,\cdots,u_N^2)Q_\lambda({\bf x}),\nonumber
\eea
which is the same as in \cite{PS,NVT}.

In special case $t=0$, \bea
|\Psi_N(u_1,\cdots,u_N)\rangle&=&(u_1\cdots u_N)^{-M_2+M_1+1}\sum_{i=0}^N\sum_{k_1,\cdots,k_i}(u_{\{k\}})^{-2M_1-2}\prod_{j\neq k_i}\frac{u_j^2}{u_j^2-u_{k_i}^2}\nonumber\\
&&\sum_{\lambda,\mu}S_{\lambda}(u_{k_1}^2,\cdots,u_{k_i}^2)S_\mu(u_1^2,\cdots,u_N^2)S_{[\lambda,\mu]}({\bf x},{\bf y}),\nonumber
\eea
which is the same as in \cite{Wang1}.
\section{Vertex operators and topological strings on the conifold}\label{sect5}

The A-model topological string partition function on the conifold is given by
\be
Z_{conifold}^{top}(z,t)=\prod_{n=1}^\infty\frac{(1-tz^n)^n}{(1-z^n)^n},
\ee
which equals the generating function of weighted plane partitions\cite{MV} and turns into the MacMahon functions when $t=0$. It is known that $Z_{conifold}^{top}(z,t)$  can be written as a fermionic correlator involving the vertex operators
\[
\tilde{\Gamma}^-(z)=e^{\xi_t({\bf x},z)},\quad \tilde{\Gamma}^+(z)=e^{\xi(\tilde{\partial}_{\bf x},z^{-1})},
\]
 with a particular specialization of the values of $z=q^{\pm 1/2},q^{\pm 3/2},q^{\pm 5/2}\cdots$, i.e.,
 \be
Z_{conifold}^{top}(z,t)=\langle 0|\prod_{m=1}^\infty  \tilde{\Gamma}^+(z^{-m+1/2})\prod_{m=1}^\infty \tilde{\Gamma}^-(z^{m-1/2})|0\rangle.
\ee
In the following, we will give that $Z_{conifold}^{top}(z,t)$ can be obtained from the vertex operators $\Gamma^{\pm}_i(t)$, $i=1,2$.
\begin{lem}\label{lemmazc}
The following relation holds
\be
\langle 0|\prod_{m=1}^\infty \Gamma_2^-(z^{m-1/2})\Gamma_1^-(z^{m-1/2})|0\rangle=\prod_{n\geq 1}\frac{(1-z^n)^n}{(1-tz^n)^n}\frac{(1-t^2z^n)^n}{(1-tz^n)^n}.
\ee
\end{lem}
\begin{proof}
From
\[
\Gamma_1^-(w)=e^{\xi_t({\bf x}, w)}e^{-\xi_t(\tilde{\partial}_{\bf y}, w)},
\]
\[
\Gamma_2^-(z)=e^{\xi_t({\bf y}, z)}e^{-\xi_t(\tilde{\partial}_{\bf x}, z)},
\]
and
\[
e^{-\xi_t(\tilde{\partial}_{\bf x}, z)}e^{\xi_t({\bf x}, w)}=\frac{1-zw}{1-tzw}\frac{1-t^2zw}{1-tzw}e^{\xi_t({\bf x}, w)}e^{-\xi_t(\tilde{\partial}_{\bf x}, z)},
\]
we get
\[
\Gamma_2^-(z)\Gamma_1^-(w)=\frac{1-zw}{1-tzw}\frac{1-t^2zw}{1-tzw}:\Gamma_2^-(z)\Gamma_1^-(w):
\]
where $::$ is the normal order defined as usual. Using this formula step by step, we get the conclusion.
\end{proof}
Basing on the above preparation, we can draw the following interesting proposition.
\begin{prp}
The A-model topological string partition function on the conifold equals
\begin{eqnarray}
&&Z_{conifold}^{top}(z,t^2)=\prod_{n=1}^\infty\frac{(1-t^2z^n)^n}{(1-z^n)^n}\nonumber\\
&=&\langle 0|\prod_{m=1}^\infty \Gamma_2^+(z^{-m+1/2})\Gamma_1^+(z^{-m+1/2})\prod_{m=1}^\infty \Gamma_2^-(z^{m-1/2})\Gamma_1^-(z^{m-1/2})|0\rangle\label{zcon}
\end{eqnarray}
\end{prp}
\begin{proof}
Since
\bea
\Gamma_i^+(z)\Gamma_i^-(w)&=&\frac{1}{1-w/z}\Gamma_i^-(w)\Gamma_i^+(z),\quad i=1,2,\nonumber\\
\Gamma_i^+(z)\Gamma_j^-(w)&=&\Gamma_i^j(w)\Gamma_i^+(z),\quad i,j=1,2,\ i\neq j.\nonumber
\eea
Then the right hand side of (\ref{zcon}) equals
\[
\frac{(1-tz^n)^n}{(1-z^n)^n}\frac{(1-tz^n)^n}{(1-z^n)^n}\langle 0|\prod_{m=1}^\infty \Gamma_2^-(z^{m-1/2})\Gamma_1^-(z^{m-1/2})|0\rangle
\]
By lemma (\ref{lemmazc}), we get the conclusion.

\end{proof}

\section*{Acknowledgements}
The authors gratefully acknowledge the support of Professors Ke Wu, Zi-Feng Yang, Shi-Kun Wang.
Chuanzhong Li is supported by the National Natural Science Foundation
of China under Grant No. 11571192 and K. C. Wong Magna Fund in Ningbo University.

\end{document}